  \newcommand{\XXVIII}{GTP28}
  \newcommand{\XLII}{GTP42}
  \newcommand{\XLIII}{GTP43}
  \newcommand{\XLIV}{GTP44}

  \documentclass[10pt]{article}
  \usepackage{amsmath,amsthm,amsfonts,amssymb,latexsym,mathtools,euscript,url,pifont}
  \usepackage[numbers]{natbib}
  \usepackage[pdfpagemode=UseNone,pdfstartview=FitH]{hyperref}
  \newcommand{\Extra}[1]{}

\allowdisplaybreaks[4]

\newcommand*{\Subsection}[1]{\subsection*{#1}\addcontentsline{toc}{subsection}{#1}}

\newcommand*{\st}{\mathrel{|}}
\newcommand*{\dd}{\,\mathrm{d}}
\newcommand*{\K}{\mathcal{K}}
\newcommand*{\FFF}{\mathcal{F}}

\newcommand*{\TTT}{\mathcal{T}}
\newcommand*{\cemetery}{\text{\ding{61}}}

\newcommand*{\Mu}{\mathrm{M}}

\DeclareMathOperator{\III}{\boldsymbol{1}}
\DeclareMathOperator{\dom}{dom}

\DeclareMathOperator{\Var}{Var}
\DeclareMathOperator{\Cov}{Cov}

\newcommand*{\bbbp}{\mathbb{P}}
\DeclareMathOperator{\Prob}{\bbbp}
\DeclareMathOperator{\UpProb}{\overline{\bbbp}}

\newcommand{\pretend}[2]{\smash{\mathrlap{#1}}\phantom{#2}}
\DeclareMathOperator{\UpProbI}{{\textstyle\pretend{\overline{\bbbp}}{\bbbp}^{\it I}}}

\newcommand*{\bbbe}{\mathbb{E}}
\DeclareMathOperator{\Expect}{\bbbe}
\DeclareMathOperator{\UpExpect}{\overline{\bbbe}}

\DeclareMathOperator{\UpExpectI}{\textstyle\pretend{\overline{\bbbe}}{\bbbe}^{\it I}}

\DeclareMathOperator{\EEE}{\mathcal{E}}
\DeclareMathOperator{\LLL}{\mathcal{L}}

\newcommand*{\bbbr}{\mathbb{R}}

\theoremstyle{plain}
\newtheorem{theorem}{Theorem}[section]

\newtheorem{corollary}[theorem]{Corollary}
\newtheorem{lemma}[theorem]{Lemma}
\theoremstyle{definition}
\newtheorem{remark}[theorem]{Remark}

\title{Towards a probability-free theory of continuous martingales}
\author{Vladimir Vovk and Glenn Shafer\\
  \texttt{{\rm\{}volodya.vovk,glennrayshafer{\rm\}}{\rm@}gmail.com}}

\begin{document}
  \maketitle

  \begin{abstract}
    Without probability theory, we define classes of supermartingales, martingales, and semimartingales
    in idealized financial markets with continuous price paths.
    This allows us to establish probability-free versions of a number of standard results in martingale theory,
    including the Dubins--Schwarz theorem, the Girsanov theorem, and results concerning the It\^o integral.
    We also establish the existence of an equity premium and a CAPM relationship in this probability-free setting.

      \bigskip

      \noindent
      The version of this paper at
      \href{http://probabilityandfinance.com/articles/index.html#45}{http://probabilityandfinance.com}
      (Working Paper 45)
      is updated most often.
  \end{abstract}

\section{Introduction}

We consider a financial market in which a finite number of securities with continuous price paths are traded.
We do not make any stochastic assumptions,
and our basic definitions are in the spirit of what we call game-theoretic probability (see, e.g., \cite{Shafer/Vovk:2001}).
This theory's key notion, probability-free superhedging, has been formalized in different ways
in the case of continuous time.
The definition suggested in \cite{\XXVIII} is very cautious.
Perkowski and Pr\"omel \cite{Perkowski/Promel:2016-local} propose a broader definition,
making superhedging easier.
In this paper we propose an even broader definition of superhedging
and use it to define concepts such as continuous martingale, nonnegative supermartingale, etc.
This allows us to derive probability-free versions of several standard results of martingale theory.
  In particular, we will give a simple proof of the Girsanov theorem stated
  (and proved in a roundabout way)
  in \cite{\XLIV}.
At the end of the paper we use our results to give a probability-free treatment of the equity premium and CAPM.

\section{Supermartingales, martingales, and semimartingales}
\label{sec:martingales}

Our model of the financial market contains $J^*$ traded securities whose price paths are denoted $S_1,\ldots,S_{J^*}$;
these are continuous functions $S_j:[0,\infty)\to\bbbr$, $j=1,\ldots,J^*$.
Apart from the price paths of traded securities,
our model will also contain other paths (``information paths'')
$S_{J^*+1},\ldots,S_{J}$
reflecting the ``side information'' (such as the economy's or individual companies' fundamentals)
available to the traders;
these functions are also assumed to be continuous.

Formally, our \emph{sample space} is the set $\Omega:=C[0,\infty)^J$ of all $J$-tuples $\omega=(S_1,\ldots,S_{J})$
of continuous functions $S_j:[0,\infty)\to\bbbr$, $j=1,\ldots,J$;
$J^*\in\{1,\ldots,J\}$ is a parameter of the market (the number of traded securities).
Each $\omega=(S_1,\ldots,S_{J})\in\Omega$ is identified with the function
$\omega:[0,\infty)\to(0,\infty)^{J}$
defined by $\omega(t):=(S_1(t),\ldots,S_{J}(t))$, $t\in[0,\infty)$.
We equip $\Omega$ with the $\sigma$-algebra $\FFF$
generated by the functions $\omega\in\Omega\mapsto\omega(t)$, $t\in[0,\infty)$
(i.e., the smallest $\sigma$-algebra making them measurable).
We often consider subsets of $\Omega$ and functions on $\Omega$
(which we often call \emph{functionals})
that are measurable with respect to $\FFF$.

A \emph{random vector} is an $\FFF$-measurable function of the type $\Omega\to\bbbr^d$ for some $d\in\{1,2,\ldots\}$,
and an \emph{extended random variable} is an $\FFF$-measurable function of the type $\Omega\to[-\infty,\infty]$.
A \emph{stopping time} is an extended random variable $\tau$ taking values in $[0,\infty]$ such that,
for all $\omega$ and $\omega'$ in $\Omega$,
\begin{equation*}
  \left(
    \omega|_{[0,\tau(\omega)]}
    =
    \omega'|_{[0,\tau(\omega)]}
  \right)
  \Longrightarrow
  \left(
    \tau(\omega)=\tau(\omega')
  \right),
\end{equation*}
where $f|_A$ stands for the restriction of $f$ to the intersection of $A$ and $f$'s domain.
A random vector $X$ is said to be \emph{$\tau$-measurable},
where $\tau$ is a stopping time,
if,
for all $\omega$ and $\omega'$ in $\Omega$,
\begin{equation*}
  \left(
    \omega|_{[0,\tau(\omega)]}
    =
    \omega'|_{[0,\tau(\omega)]}
  \right)
  \Longrightarrow
  \left(
    X(\omega)=X(\omega')
  \right).
\end{equation*}
A \emph{process} is a function $X:[0,\infty)\times\Omega\to[-\infty,\infty]$.
The process is \emph{adapted} if,
for all $\omega$ and $\omega'$ in $\Omega$ and all $t\in[0,\infty)$,
\begin{equation*}
  \left(
    \omega|_{[0,t]}
    =
    \omega'|_{[0,t]}
  \right)
  \Longrightarrow
  \left(
    X_t(\omega)=X_t(\omega')
  \right).
\end{equation*}
Our definitions are in the spirit of the Galmarino test
(see, e.g., \cite{Dellacherie/Meyer:1978-local}, IV.100);
in later sections they will often make checking that various $[0,\infty]$-valued extended random variables are stopping times
straightforward.
As customary in probability theory,
we will often omit explicit mention of $\omega\in\Omega$
when it is clear from the context.

A \emph{simple trading strategy} $G$ is a pair $((\tau_1,\tau_2,\ldots),(h_1,h_2,\ldots))$,
where:
\begin{itemize}
\item
  $\tau_1\le\tau_2\le\cdots$ is an increasing sequence of stopping times
  such that, for each $\omega\in\Omega$,
  $\lim_{n\to\infty}\tau_n(\omega)=\infty$;
\item
  for each $n=1,2,\ldots$, $h_n$ is a bounded $\tau_{n}$-measurable $\bbbr^{J^*}$-valued random vector.
\end{itemize}
The \emph{simple capital process} $\K^{G,c}$
corresponding to a simple trading strategy $G$
and \emph{initial capital} $c\in\bbbr$ is defined by
\begin{multline}\label{eq:simple-capital-1}
  \K^{G,c}_t(\omega)
  :=
  c
  +
  \sum_{n=1}^{\infty}
  h_n(\omega)\cdot
  \bigl(
    \omega^*(\tau_{n+1}\wedge t)-\omega^*(\tau_n\wedge t)
  \bigr),\\
  t\in[0,\infty),
  \enspace
  \omega\in\Omega,
\end{multline}
where ``$\cdot$'' stands for dot product in $\bbbr^{J^*}$,
$\omega^*:=(S_1,\ldots,S_{J^*})$ consists of the first $J^*$ components of $\omega$,
and the zero terms in the sum are ignored
(which makes the sum finite for each $t$).
We will refer to the $j$th component of $h_n$ as the \emph{bet} on $S_j$
over $(\tau_i,\tau_{i+1}]$.
Notice that (a) simple trading strategies trade only in the first $J^*$ $S_j$s
(corresponding to the traded securities)
but the stopping times and bets can depend on all $J$ $S_j$s,
and (b) expression \eqref{eq:simple-capital-1} implicitly assumes zero interest rates
(this assumption will be removed in Section~\ref{sec:numeraire}).
All simple capital processes have continuous paths and are adapted.

Let us say that a class $\mathcal{C}$ of nonnegative processes is \emph{$\liminf$-closed}
if the process
\begin{equation}\label{eq:PP}
  X_t(\omega)
  :=
  \liminf_{k\to\infty}
  X^k_t(\omega)
\end{equation}
is in $\mathcal{C}$ whenever
each process $X^k$ is in $\mathcal{C}$.
A process $X$ 
is a \emph{nonnegative supermartingale} if it belongs to the smallest $\liminf$-closed class of nonnegative processes
containing all nonnegative simple capital processes.
Intuitively, nonnegative supermartingales are nonnegative capital processes
(in fact, they can lose capital as the approximation is in the sense of $\liminf$).

\begin{remark}\label{rem:supermartingales}
  An equivalent definition of the class $\mathcal{C}$ of nonnegative supermartingales
  can be given using transfinite induction on the countable ordinals $\alpha$
  (see, e.g., \cite{Dellacherie/Meyer:1978-local}, 0.8).
  Namely, define $\mathcal{C}^{\alpha}$ as follows:
  \begin{itemize}
  \item
    $\mathcal{C}^0$ is the class of all nonnegative simple capital processes;
  \item
    for $\alpha>0$,
    $X\in\mathcal{C}^{\alpha}$ if and only if there exists a sequence $X^1,X^2,\ldots$
    of processes in $\mathcal{C}^{<\alpha}:=\cup_{\beta<\alpha}\mathcal{C}^{\beta}$
    such that
    \eqref{eq:PP} holds.
  \end{itemize}
  It is easy to check that the class of all nonnegative supermartingales
  is the union of the nested family $\mathcal{C}^{\alpha}$
  over all countable ordinals $\alpha$.
  The \emph{rank} of a nonnegative supermartingale $X$
  is defined to be the smallest $\alpha$ such that $X\in\mathcal{C}^{\alpha}$;
  in this case we will also say that $X$ is \emph{of rank $\alpha$}.
\end{remark}

We call a subset of $[0,\infty)\times\Omega$ a \emph{property of $t$ and $\omega$}.
We say that a property $E$ of $t$ and $\omega$ holds \emph{quasi-always} (q.a.)\
if there exists a nonnegative supermartingale $X$ such that $X_0=1$ and,
for all $t\in[0,\infty)$ and $\omega\in\Omega$,
\begin{equation*}
  (t,\omega)\notin E
  \Longrightarrow
  X_t(\omega) = \infty.
\end{equation*}
(This implies that the complement of $E$ is evanescent in the sense of its projection onto $\Omega$
having zero game-theoretic upper probability,
as defined in Section~\ref{sec:Dubins-Schwarz} below.)

\begin{lemma}\label{lem:qa}
  If each property in a countable set of properties of $t$ and $\omega$
  holds quasi-always,
  their intersection also holds quasi-always.
\end{lemma}

\begin{proof}
  It suffices to notice that a countable convex mixture of nonnegative supermartingales is a nonnegative supermartingale.
\end{proof}

A sequence of processes $X^k$ converges to a process $X$ \emph{uniformly on compacts quasi-always (ucqa)}
if the property
\begin{equation}\label{eq:property}
  \lim_{k\to\infty}
  \sup_{s\in[0,t]}
  \left|
    X^k_s(\omega)-X_s(\omega)
  \right|
  =
  0
\end{equation}
of $t$ and $\omega$ holds quasi-always.
If continuous $X^k$ converge ucqa to $X$,
we can consider the limit $X$ to be continuous as well;
to make this precise, we will extend the notion of a continuous process.

Let $\cemetery$ (the \emph{cemetery state}) be any element outside the real line $\bbbr$.
Adapted processes $X:[0,\infty)\times\Omega\to\bbbr\cup\{\cemetery\}$
are defined analogously to adapted $[-\infty,\infty]$-valued processes.
Let us say that an adapted process $X:[0,\infty)\times\Omega\to\bbbr\cup\{\cemetery\}$
is a \emph{continuous process} if
\begin{itemize}
\item
  it takes values in $\bbbr$ quasi-always;
\item
  for each $\omega\in\Omega$,
  \begin{itemize}
  \item
    the set of $t\in[0,\infty)$ for which $X_t(\omega)\in\bbbr$
    contains 0 and is
    connected;
  \item
    $X_t(\omega)$ is continuous as function of $t$ in this set.
  \end{itemize}
\end{itemize}
(Even though an object that qualifies as a continuous process by this definition does not qualify as a process by our earlier definition,
we will sometimes call it a process when there is no danger of confusion.)
The \emph{effective domain} of a continuous process $X$ is defined to be
\[
  \dom X := \{(t,\omega) \st X_t(\omega)\in\bbbr\}.
\]
A class $\mathcal{C}$ of continuous processes is \emph{$\lim$-closed}
if it contains every continuous process $X$
for which there exists a sequence $X^k$ of continuous processes in $\mathcal{C}$ such that:
\begin{itemize}
\item
  $\dom X\subseteq\dom X^k$ for each $k$;
\item
  \eqref{eq:property} holds for each $(t,\omega)\in\dom X$.
\end{itemize}
A continuous process is a \emph{continuous martingale}
if it is an element of the smallest $\lim$-closed class of continuous processes
that contains all simple capital processes.
The \emph{rank} of a continuous martingale is defined as in Remark~\ref{rem:supermartingales}.

The following lemma will be useful in establishing that various specific functions $\tau:\Omega\to[0,\infty]$
are stopping times.

\begin{lemma}\label{lem:Sigma}
  For any continuous process $X$, the function $\Sigma_X:\Omega\to[0,\infty]$ defined by
  \[
    \Sigma_X(\omega)
    :=
    \inf\{t\in[0,\infty)\st X_t(\omega)=\cemetery\}
  \]
  is measurable.
\end{lemma}

\begin{proof}
  It suffices to notice that, for each $t\in[0,\infty)$,
  \[
    \{\omega\st\Sigma_X(\omega)\le t\}
    =
    \bigcap_{\epsilon}
    \:
    \{\omega\st X_{t+\epsilon}(\omega)=\cemetery\},
  \]
  where $\epsilon$ ranges over the positive rational numbers.
\end{proof}

Lemma~\ref{lem:Sigma} does not claim that $\Sigma_X$ itself is a stopping time
(it is not in many interesting cases).

A continuous process $A$ is a \emph{finite variation continuous process}
if $A_0(\omega)=0$ and the total variation of the function $s\in[0,t]\mapsto A_s(\omega)$ is finite
for all $(t,\omega)\in\dom A$.
A continuous process $X$ is a \emph{continuous semimartingale}
if there exist a continuous martingale $Y$ and a finite variation continuous process $A$
such that $\dom X=\dom Y=\dom A$ and $X=Y+A$.
We will call such a decomposition $X=Y+A$ of $X$ the \emph{standard decomposition},
where the article ``the'' is justified by its uniqueness:
see Corollary~\ref{cor:decomposition} below.

\Subsection{Comparison with measure-theoretic probability}

The motivation for our terminology is the analogy with measure-theoretic probability.
In this subsection we suppose that $S_1,\ldots,S_{J^*}$ are continuous local martingales
on a measure-theoretic probability space with a given filtration,
whereas $S_{J^*+1},\ldots,S_J$ are continuous adapted stochastic processes.
Each simple capital process is a local martingale.
Since each nonnegative local martingale is a supermartingale
(\cite{Revuz/Yor:1999}, p.~123),
nonnegative simple capital processes are supermartingales.
By the Fatou lemma, $\liminf_k X^k$ is a supermartingale
whenever $X^k$ are nonnegative supermartingales:
\begin{equation*}
  \Expect
  \left(
    \liminf_k X^k_t\st\FFF_s
  \right)
  \le
  \liminf_k
  \Expect(X^k_t\st\FFF_s)
  \le
  \liminf_k X^k_s
  \quad
  \text{a.s.},
\end{equation*}
where $0\le s<t$.
Therefore, our definition gives a subset of the set of all nonnegative measure-theoretic supermartingales.
(We are using the definitions of measure-theoretic supermartingales and martingales
that do not impose any continuity conditions, as in \cite{Revuz/Yor:1999}, Definition~II.1.1.)

Let us now check that continuous martingales $X$, as defined above (but with $\cemetery$ replaced by, say, 0),
are continuous local martingales in the measure-theoretic setting.
Since simple capital processes are continuous local martingales
and a limit of a sequence of continuous local martingales that converge
in probability uniformly on compact time intervals
is always a continuous local martingale (\cite{Cherny:2006}, Theorem 3.1),
it suffices to apply transfinite induction on the rank of~$X$.
Notice that we can make all sample paths of $X$ continuous by changing it on a set of measure zero.

\section{Conservatism of continuous martingales}

In this section we derive a key technical result of this paper showing that adding a continuous martingale to our market
as a new traded security (in addition to the basic price paths $S_1,\ldots,S_{J^*}$)
is its ``conservative extension,'' to use a logical term:
it does not increase the supply of nonnegative supermartingales and continuous martingales.
But we start from a simpler property:
adding a continuous process as a new piece of side information
(in addition to the basic information paths $S_{J^*+1},\ldots,S_{J}$)
is a conservative extension.
This property is stated as Theorem~\ref{thm:side-conservative-super} for nonnegative supermartingales
and as Theorem~\ref{thm:side-conservative} for continuous martingales.

\begin{theorem}\label{thm:side-conservative-super}
  If $H$ is a continuous process,
  adding $H$ to the market as side information
  does not add any new nonnegative supermartingales.
\end{theorem}

Before proving Theorem~\ref{thm:side-conservative-super},
we will give its more detailed and explicit statement.
Let $\Omega':=C[0,\infty)^{J+1}$ be the sample space of the new market $(S_1,\ldots,S_J,H')$,
where $H':[0,\infty)\to\bbbr$ is the new information path.
Let $Y$ be a nonnegative supermartingale in the new market.
Any expression of the form $Y_t(z)$,
where $t\in[0,\infty)$ and $z:[0,\infty)\to(\bbbr\cup\{\cemetery\})^{J+1}$,
is understood to be $\infty$ if $z$ ever takes a value containing $\cemetery$ over $[0,t]$,
and to be $Y_t(z')$ otherwise, where $z'$ is any element of $\Omega'$ such that $z'|_{[0,t]}=z|_{[0,t]}$
(we will only be interested in cases where there is no dependence on such $z'$).
The theorem says that
\begin{multline}\label{eq:Y}
  Y'_t(\omega)
  =
  Y'_t(S_1,\ldots,S_J)\\
  :=
  Y_t(S_1,\ldots,S_J,H(S_1,\ldots,S_J))
  =
  Y_t(\omega,H(\omega)),
\end{multline}
where $H(S_1,\ldots,S_J)$ is the function $t\in[0,\infty)\mapsto H_t(S_1,\ldots,S_J)$,
is a nonnegative supermartingale (in the old market).

\begin{proof}[Proof of Theorem~\ref{thm:side-conservative-super}]
  Suppose $Y$ is a nonnegative supermartingale in the new market $(S_1,\ldots,S_{J},H')$.
  It suffices to consider the case where $Y$ is a simple capital process.
  Indeed, for the other nonnegative supermartingales $Y$ we can use transfinite induction:
  if
  \begin{equation}\label{eq:induction-1}
    Y_t(S_1,\ldots,S_J,H')
    =
    Y_t(\omega,H')
    =
    \liminf_{k\to\infty}
    Y^k_t(\omega,H')
  \end{equation}
  for some nonnegative supermartingales
  $
    Y^k_t=Y^k_t(\omega,H')
  $
  of lower ranks,
  the inductive assumption will imply that
  \begin{equation}\label{eq:induction-2}
    Y'_t(\omega)
    :=
    Y_t
    \bigl(
      \omega,
      H(\omega)
    \bigr)
    =
    \liminf_{k\to\infty}
    Y_t^k
    \bigl(
      \omega,
      H(\omega)
    \bigr)
  \end{equation}
  is a nonnegative supermartingale in the old market.
  (The equality $=$ in~\eqref{eq:induction-2} is easy to check both in the case $(t,\omega)\in\dom H$,
  when it follows from \eqref{eq:induction-1},
  and in the case $(t,\omega)\notin\dom H$,
  when that equality becomes $\infty=\infty$.)

  Therefore, we are given a nonnegative simple capital process $Y$ in the new market,
  and our goal is to prove that \eqref{eq:Y} is a nonnegative supermartingale in the old market.
  Let $(\tau_1,\tau_2,\ldots)$ and $(h_1,h_2,\ldots)$ be the corresponding stopping times and bets
  ($\bbbr^{J^*}$-valued random vectors).
  Spelling out the dependence of $H$ on the basic paths,
  we can consider $\tau_n$ and $h_n$ to be stopping times and random vectors
  on the sample space $\Omega$ of the old market.
  Namely, we define the counterparts of $\tau_n$ and $h_n$ in the old market as follows:
  \begin{itemize}
  \item
    if there exist $t\in[0,\infty)$ and $H'\in C[0,\infty)$
    such that
    \begin{equation}\label{eq:condition}
      H'|_{[0,t]} = H(\omega)|_{[0,t]}
      \text{ and }
      \tau_n(\omega,H')\le t,
    \end{equation}
    set $\tau'_n(\omega):=\tau_n(\omega,H')$;
  \item
    otherwise, set $\tau'_n(\omega):=\infty$;
  \item
    in any case, set
    \[
      h'_n(\omega)
      :=
      \begin{cases}
        h_n(\omega,H') & \text{if $\tau'_n(\omega)<\infty$}\\
        0 & \text{otherwise},
      \end{cases}
    \]
    where $H'\in C[0,\infty)$ satisfies \eqref{eq:condition} for some $t$.
  \end{itemize}
  There is no dependence on the choice of $H'$,
  the $\tau'_n$ are stopping times, and the $h_n$ are $\tau_n$-measurable random variables in the old market.
  All these statements are trivial apart from the measurability of $\tau'_n$,
  but the latter
  can be easily deduced from the fact that our $\sigma$-algebra on $\Omega$
  coincides with the Borel $\sigma$-algebra for the uniform topology on the basic paths.

  Define $Y''$ as the simple capital process in the old market corresponding to the initial capital $Y_0$
  and the simple trading strategy $((\tau'_1,\tau'_2,\ldots),(h'_1,h'_2,\ldots))$.
  Define $Y'''$ as $Y''$ stopped when it reaches 0 (if it ever does;
  notice that $Y''$ is not guaranteed to be nonnegative even though $Y$ is nonnegative).
  Let $X$ be any nonnegative supermartingale with $X_0=1$ that is infinite outside $\dom H$.
  It remains to notice that $Y'$ is the limit of $Y'''+X/k$ as $k\to\infty$.
\end{proof}

\begin{theorem}\label{thm:side-conservative}
  If $H$ is a continuous process,
  adding $H$ to the market as side information
  does not add new continuous martingales.
\end{theorem}

To give a more explicit statement of Theorem~\ref{thm:side-conservative},
we now define an expression of the form $Y_t(z)$,
where $Y$ is a continuous martingale in the new market,
$t\in[0,\infty)$, and $z:[0,\infty)\to(\bbbr\cup\{\cemetery\})^{J+1}$,
as $\cemetery$ if $z$ takes value outside $\bbbr^{J+1}$ over $[0,t]$,
and as $Y_t(z')$ otherwise, where $z'$ is any element of $\Omega'$ such that $z'|_{[0,t]}=z|_{[0,t]}$.
The theorem says that \eqref{eq:Y} is a continuous martingale in the old market.

\begin{proof}[Proof of Theorem~\ref{thm:side-conservative}]
  The proof will go along the lines of the proof of Theorem~\ref{thm:side-conservative-super}.
  Suppose $Y$ is a continuous martingale in the new market $(S_1,\ldots,S_{J},H')$.
  First we reduce the problem to the case where $Y$ is a simple capital process.
  The inductive step \eqref{eq:induction-1}--\eqref{eq:induction-2} now is:
  if
  \begin{equation}\label{eq:induction-1a}
    \lim_{k\to\infty}
    \sup_{s\in[0,t]}
    \left|
      Y^k_s(\omega,H')
      -
      Y_s(\omega,H')
    \right|
    =
    0
  \end{equation}
  for all $(t,(\omega,H'))\in\dom Y$ for some continuous martingale $Y$
  and some continuous martingales $Y^k$ of lower ranks than $Y$ in the new market,
  the inductive assumption
  ($Y^k_t(\omega,H(\omega))$ being a continuous martingale in the old market for all $k$)
  will imply that $Y'_t(\omega):=Y_t(\omega,H(\omega))$
  is a continuous martingale in the old market.
  To establish the last statement it suffices to prove that
  \begin{equation}\label{eq:induction-2a}
    \lim_{k\to\infty}
    \sup_{s\in[0,t]}
    \left|
      Y_s^k(\omega,H(\omega))
      -
      Y_s(\omega,H(\omega))
    \right|
    =
    0
  \end{equation}
  holds for all $(t,\omega)\in\dom Y'\subseteq\dom H$ and that $\dom Y'$ holds quasi-always.
  For $(t,\omega)\in\dom H$, $(t,\omega)\in\dom Y'$ means that
  $(t,(\omega,H'))\in\dom Y$ for any $H'\in C[0,\infty)$ such that $H'|_{[0,t]}=H(\omega)|_{[0,t]}$,
  and so the equality \eqref{eq:induction-2a} follows from \eqref{eq:induction-1a}.
  Next we prove that $\dom Y'$ holds quasi-always.
  Let $X$ be a nonnegative supermartingale in the new market such that $X_0=1$
  and $X_t(\omega,H')=\infty$ for all $(t,(\omega,H'))\notin\dom Y$.
  By Theorem~\ref{thm:side-conservative-super}, $X_t(\omega,H(\omega))$ is a nonnegative supermartingale in the old market,
  and this nonnegative supermartingale witnesses that $\dom Y'$ holds quasi-always.

  It remains to consider the case where $Y$ is a simple capital process,
  and our goal is to prove that \eqref{eq:Y} is a continuous martingale in the old market.
  This follows from the effective domain of \eqref{eq:Y} being $\dom H$
  and \eqref{eq:Y} being equal inside $\dom H$ to the simple capital process $Y''$ in the old market
  constructed in the proof of Theorem~\ref{thm:side-conservative-super}
  (the construction also works without the assumption that $Y$ is nonnegative).
\end{proof}

Now we prove the analogues of Theorems~\ref{thm:side-conservative-super} and~\ref{thm:side-conservative}
for the price paths of traded securities;
Theorem~\ref{thm:conservative-super} covers nonnegative supermartingales,
and Theorem~\ref{thm:conservative} is about continuous martingales.

\begin{theorem}\label{thm:conservative-super}
  Let $X$ be a continuous martingale.
  Consider the extended market in which $X$ is traded as a new security.
  Any nonnegative supermartingale $Y$ in the new market
  is a nonnegative supermartingale in the old market.
\end{theorem}

Let $\Omega':=C[0,\infty)^{J+1}$ be the sample space of the new market $(X',S_1,\ldots,S_J)$,
where $X':[0,\infty)\to\bbbr$ is the price path of the new security
(we always list price paths before information paths).
Let $Y$ be a nonnegative supermartingale in the new market.
Our understanding of expressions of the form $Y_t(z)$ is the same
as for Theorem~\ref{thm:side-conservative-super}.
Theorem~\ref{thm:conservative-super} says that
\begin{multline}\label{eq:Y-bis}
  Y'_t(\omega)
  =
  Y'_t(S_1,\ldots,S_J)\\
  :=
  Y_t(X(S_1,\ldots,S_J),S_1,\ldots,S_J)
  =
  Y_t(X(\omega),\omega)
\end{multline}
(cf.\ \eqref{eq:Y}) is a nonnegative supermartingale in the old market.

\begin{proof}[Proof of Theorem~\ref{thm:conservative-super}]
  Suppose $Y$ is a nonnegative supermartingale in the new market $(X',S_1,\ldots,S_{J})$.
  The argument given in Theorem~\ref{thm:side-conservative-super}
  (cf.\ \eqref{eq:induction-1}--\eqref{eq:induction-2})
  shows that it suffices to consider the case where $Y$ is a simple capital process.
  
  Therefore, we are given a nonnegative simple capital process $Y$ in the new market,
  and our goal is to prove that it is a nonnegative supermartingale in the old market.
  The continuous martingale $X$ enters the picture in two places:
  first, it is used when defining the stopping times and bets,
  and second through the increments of $X$ entering the increments of the capital.
  (The basic price paths $S_j$, $j=1,\ldots,J^*$, play both roles,
  while the basic information paths $S_j$, $j=J^*+1,\ldots,J$, only play the first role.)
  We can use different continuous martingales $X'$ and $X''$ for these two roles
  ($X'$ for the first and $X''$ for the second),
  and we prove the required statement for any $X'$ and $X''$, without assuming $X'=X''$.
  Theorem~\ref{thm:side-conservative-super} says that adding $X'$ to the market
  does not change the class of nonnegative supermartingales;
  therefore, we will ignore $X'$.
  The rest of the proof will be by transfinite induction on the rank of $X''$,
  which from now on we denote simply as $X$.

  Let $(\tau_1,\tau_2,\ldots)$ and $(h_1,h_2,\ldots)$ be $Y$'s stopping times and bets
  ($\bbbr^{J^*+1}$-valued random vectors).
  First we suppose that $X$ is a simple capital process;
  let its stopping times and bets (this time $\bbbr^{J^*}$-valued random vectors)
  be $(\tau'_1,\tau'_2,\ldots)$ and $(h'_1,h'_2,\ldots)$.
  For each $\omega\in\Omega$ (i.e., each $\omega$ in the sample space for the old market),
  we can rearrange $\tau_1(\omega),\tau_2(\omega),\ldots,\tau'_1(\omega),\tau'_2(\omega),\ldots$
  into an increasing sequence $\tau''_1(\omega)\le\tau''_2(\omega)\le\cdots$
  and define the bet on $S_j$, $j\in\{1,\ldots,J^*\}$, over the interval $(\tau''_n,\tau''_{n+1}]$
  (taken at time $\tau''_n$)
  as $h_Y^{0}h_X^j+h_Y^j$,
  where $h_X^j$ (the $j$th component of one of the $h_i$) is the bet of $X$ on $S_j$ over that interval,
  $h_Y^{0}$ is the bet of $Y$ on $X'$ over that interval,
  and $h_Y^j$ is the bet of $Y$ on $S_j$ over that interval.
  These bets and the stopping times $\tau''_n$ give rise to $Y$
  as simple capital process in the old market;
  its effective domain is the whole of $\Omega$.

  Finally we apply transfinite induction on the rank of $X$.
  Suppose that $X$ is a continuous martingale of rank $\alpha$
  and that $X=\lim_{k\to\infty}X^k$ on $\dom X$ for continuous martingales $X^k$ of ranks less than $\alpha$.
  By the inductive assumption, for each $k$,
  the simple capital process $Y$ (determined by $(\tau_1,\tau_2,\ldots)$ and $(h_1,h_2,\ldots)$)
  applied to the old market extended by adding $X^k$
  gives a nonnegative supermartingale $Y^k_t(\omega):=Y_t(X^k(\omega),\omega)$.
  The nonnegative supermartingales $Y^k$ will converge to $Y'_t(\omega):=Y_t(X(\omega),\omega)$ inside $\dom X$
  (even in the sense of $\lim$, let alone in the sense of $\liminf$).
  Let $X'$ be a nonnegative supermartingale such that $X'_0=1$ and $X'_t(\omega)=\infty$ when $(t,\omega)\notin\dom X$.
  Then the nonnegative supermartingales $Y^k+X'/k$ will converge to $Y'$ everywhere as $k\to\infty$,
  and so $Y'$ is a nonnegative supermartingale as well.
\end{proof}

\begin{theorem}\label{thm:conservative}
  Let $X$ be a continuous martingale.
  Consider the extended market in which $X$ is traded as a new security.
  Any continuous martingale $Y$ in the new market is a continuous martingale in the old market.
\end{theorem}

The interpretation of the statement is the same as for Theorem~\ref{thm:side-conservative}
but with \eqref{eq:Y-bis} in place of \eqref{eq:Y}.

\begin{proof}[Proof of Theorem~\ref{thm:conservative}]
  Suppose $Y$ is a continuous martingale in the new market $(X',S_1,\ldots,S_{J})$.
  As before, it suffices to consider the case where $Y$ is a simple capital process.

  We are given a simple capital process $Y$ in the new market,
  and we need to check that it is a continuous martingale in the old market.
  As in the proof of Theorem~\ref{thm:conservative-super},
  we will assume that $X$ is only used for trading and not as information path.
  In the same proof we showed that $Y$ is a simple capital process in the old market
  when $X$ is a simple capital process.
  Next we proceed, as usual, by transfinite induction on the rank of $X$.
  Suppose that $X$ is a continuous martingale of rank $\alpha$
  and that $X=\lim_{k\to\infty}X^k$ inside $\dom X$ for continuous martingales $X^k$ of ranks less than $\alpha$.
  By the inductive assumption, for each $k$,
  $Y^k_t(\omega):=Y_t(X^k(\omega),\omega)$ is a continuous martingale.
  These continuous martingales will converge to $Y'_t(\omega):=Y_t(X(\omega),\omega)$
  uniformly on compact time intervals
  inside the event $\cap_{k=1}^{\infty}\dom Y^k\cap\dom X$.
  Since this event holds quasi-always (cf.\ Lemma~\ref{lem:qa}),
  $Y'$ is a continuous martingale.
\end{proof}

\section{It\^o integration}

In this section we combine the definitions and results of \cite{\XLII}
and the earlier paper \cite{Perkowski/Promel:2016-local}
with the conservatism of continuous processes and martingales
(Theorems~\ref{thm:side-conservative-super}--\ref{thm:conservative}).
Our exposition follows, to a large degree, \cite{\XLII}.

Let $H$ be a continuous process and $X$ be a continuous martingale.
In view of Theorems~\ref{thm:side-conservative-super}--\ref{thm:conservative},
we can assume that $H$ is a basic information path and $X$ is a basic price path,
although most of our discussion does not depend on this assumption.
(This remark is also applicable to the following sections.)
To define the It\^o integral of $H$ w.r.\ to $X$,
we need to partition time with sufficient resolution to see arbitrarily small changes in both $H$ and $X$.
A \emph{partition} is any increasing sequence $T$ of stopping times $0=T_0\le T_1\le T_2\le\cdots$
such that $\lim_{k\to\infty}T_k(\omega)>t$ holds quasi-always.
Let us say that a sequence $T^1,T^2,\ldots$ of partitions is \emph{fine} for a continuous process $Y$
if, for all $n,k\in\{1,2,\ldots\}$ and $\omega\in\Omega$,
\begin{equation*}
  \sup_{t\in[T^n_{k-1}(\omega),T^n_k(\omega)]:Y_t(\omega)\in\bbbr}
  Y_t(\omega)
  -
  \inf_{t\in[T^n_{k-1}(\omega),T^n_k(\omega)]:Y_t(\omega)\in\bbbr}
  Y_t(\omega)
  \le
  2^{-n}.
\end{equation*}
One sequence of partitions $T^1,T^2,\ldots$ that is fine for both $H$ and $X$ is
\begin{equation}\label{eq:T}
  T^n_k(\omega)
  :=
  \inf
  \left\{
    t>T^n_{k-1}(\omega)
    \st
    \left|H_t-H_{T^n_{k-1}}\right|
    \vee
    \left|X_t-X_{T^n_{k-1}}\right|
    =
    2^{-n-1}
  \right\}
\end{equation}
for $k=1,2,\ldots$,
where the equality in \eqref{eq:T} is regarded as false when $H_t=\cemetery$ or $X_t=\cemetery$.
(The fact that $T^n_k$ are stopping times follows from Lemma~\ref{lem:Sigma}.)

Given a sequence of partitions $T^n$ that is fine for a continuous process $H$ and a continuous martingale $X$, we set
\begin{equation}\label{eq:integral}
  (H\cdot X)^n_t
  :=
  \sum_{k=1}^{\infty}
  H_{T^n_{k-1}\wedge t}
  \Bigl(
    X_{T^n_{k}\wedge t}
    -
    X_{T^n_{k-1}\wedge t}
  \Bigr),
  \quad
  n=1,2,\ldots\,.
\end{equation}
for all $t\in[0,\infty)$.

\begin{theorem}\label{thm:Ito-integral}
  For any sequence of partitions $T^n$ fine for $H$ and $X$,
  $(H\cdot X)^n$ converge ucqa as $n\to\infty$.
  The limit will stay the same quasi-always if $T^n$ is replaced by another sequence of partitions fine for $H$ and $X$.
\end{theorem}

The limit whose existence is asserted in Theorem~\ref{thm:Ito-integral} will be denoted $(H\cdot X)_s$
or $\int_0^s H\dd X$ and called the \emph{It\^o integral} of $H$ w.r.\ to $X$.
Since the convergence is uniform on compact time intervals,
$(H\cdot X)_s$ is a continuous function of $s\in[0,t]$ quasi always
(which implies that $(H\cdot X)_s$ is a continuous function of $s\in[0,\infty)$
almost surely, as defined in Section~\ref{sec:Dubins-Schwarz}).

\begin{proof}[Proof of Theorem~\ref{thm:Ito-integral}]
  The first part of the theorem follows from Theorem~1 in \cite{\XLII} (proved in Section~4 of that paper)
  in combination with Theorems~\ref{thm:side-conservative-super}--\ref{thm:conservative}.
  (Theorem~1 in \cite{\XLII} is a statement about a specific sequence of partitions similar to \eqref{eq:T},
  but the argument is applicable to any fine sequence.)

  To show that the limits coincide quasi-always for two fine, for $H$ and $X$, sequences $T_1$ and $T_2$ of partitions,
  the argument in the proof of Theorem~1 in \cite{\XLII}
  should be applied to $T_1^n$ and $T_2^n$ instead of $T^n$ and $T^{n-1}$.
\end{proof}

\begin{lemma}\label{lem:stochastic-integral}
  The stochastic integral w.r.\ to a continuous martingale is a continuous martingale.
\end{lemma}

\begin{proof}
  By Theorems~\ref{thm:side-conservative} and~\ref{thm:conservative},
  \eqref{eq:integral} are continuous martingales,
  and therefore, in view of Theorem~\ref{thm:Ito-integral},
  their ucqa limit is a continuous martingale as well.
\end{proof}

Let us now check that the definition of the It\^o integral does not depend (quasi-always) on the sequence of partitions $T^n$,
provided it is interesting, in some sense, and fine enough.
First we give a very simple formal statement,
and then discuss the intuition behind it.

\begin{corollary}\label{cor:invariance}
  Let $\TTT$ be a countable set of sequences of partitions.
  There exists a continuous process $H\cdot X$ such that for any element of $\TTT$ that is fine for $H$ and $X$,
  $(H\cdot X)^n$ (defined in terms of that element) converges ucqa to that process $H\cdot X$ as $n\to\infty$.
\end{corollary}

\begin{proof}
  This follows immediately from Theorem~\ref{thm:Ito-integral} and Lemma~\ref{lem:qa}.
\end{proof}

The application of Corollary~\ref{cor:invariance} that we have in mind
is that we fix a relatively formal language
for talking about stochastic processes (such as the language of Revuz and Yor \cite{Revuz/Yor:1999}).
The language allows us to define various sequences of partitions
(perhaps referring to $H$, $X$, and the basic paths $(S_1,\ldots,S_J)$),
such as \eqref{eq:T}.
There are countably many sentences in the language,
and those of them describing sequences of partitions form a countable set
which we denote $\TTT$.
Corollary~\ref{cor:invariance} then gives us an invariant definition of It\^o integral.

\begin{remark}
  The language that we allow when defining $\TTT$
  cannot be English or the language of logic textbooks such as \cite{Mendelson:1997}:
  e.g., \cite{Mendelson:1997} contains a phrase,
  ``the least positive integer that is not denoted by an English expression
  containing fewer than 200 occurrences of symbols'' (\cite{Mendelson:1997}, p.~3, Berry's paradox),
  showing that the notion of definability can be murky
  if the language is too rich.
\end{remark}

\begin{remark}
  The use of formal languages in the foundations of probability goes back
  to at least Wald's \cite{Wald:1937-local} work on von Mises's collectives.
\end{remark}

It is easy to check that our definition of the It\^o integral $H\cdot X$ carries over verbatim
to the case where $X$ is a continuous semimartingale
and the sequence of partitions is assumed to be fine for $H$
and both components of the standard decomposition of $X$
(see Corollary~\ref{cor:decomposition}).
Alternatively, we obtain the same result (quasi-always) by setting $H\cdot X$
to $H\cdot Y + H\cdot A$, where $Y+A$ is the standard decomposition of $X$
and $H\cdot A$ is the Lebesgue--Stiltjes integral.

\section{Covariation and quadratic variation}

We start from establishing the existence of the covariation between two continuous martingales, $X$ and $Y$.
The covariation of $X$ and $Y$ can be approximated by
\begin{equation}\label{eq:covariation}
  [X,Y]^n_t
  :=
  \sum_{k=1}^{\infty}
  \left(
    X_{T^n_k\wedge t}
    -
    X_{T^n_{k-1}\wedge t}
  \right)
  \left(
    Y_{T^n_k\wedge t}
    -
    Y_{T^n_{k-1}\wedge t}
  \right),
  \quad
  n=1,2,\ldots\,.
\end{equation}
We show that the ucqa limit of $[X,Y]^n$ as $n\to\infty$ exists for fine sequences of partitions,
denote it $[X,Y]$ (or $[X,Y]_t(\omega)$ if we need to mention the arguments),
and call it the \emph{covariation} between $X$ and $Y$.

\begin{lemma}\label{lem:covariation}
  The ucqa limit of~\eqref{eq:covariation} exists for sequences of partitions
  that are fine for $X$ and $Y$.
  Moreover, it satisfies the \emph{integration by parts formula}\Extra{\ (\cite{Protter:2005-local}, Corollary~2 on p.~68)}
  \begin{equation}\label{eq:by-parts}
    X_t Y_t
    =
    (X\cdot Y)_t
    +
    (Y\cdot X)_t
    +
    [X,Y]_t
    \quad
    \text{q.a.}
  \end{equation}
\end{lemma}

\begin{proof}
  The stochastic integral $(X\cdot Y)_t=\int_0^t X_s\dd Y_s$ was defined
  in the previous section as the ucqa limit as $n\to\infty$ of
  \begin{equation*}
    (X\cdot Y)^n_t(\omega)
    :=
    \sum_{k=1}^{\infty}
    X_{T^n_{k-1}\wedge t}
    \left(
      Y_{T^n_k\wedge t}
      -
      Y_{T^n_{k-1}\wedge t}
    \right),
    \quad
    n=1,2,\ldots\,.
  \end{equation*}
  Swapping $X$ and $Y$ we obtain the analogous expression for $(Y\cdot X)_t=\int_0^t Y_s\dd X_s$.
  It is easy to check that
  \begin{equation*}
    X_t Y_t
    =
    (X\cdot Y)^n_t
    +
    (Y\cdot X)^n_t
    +
    [X,Y]^n_t.
  \end{equation*}
  Passing to the ucqa limit as $n\to\infty$ we obtain the existence of $[X,Y]$
  and the integration by parts formula~\eqref{eq:by-parts}.
\end{proof}

It is clear from \eqref{eq:by-parts} that $[X,Y]$ is a continuous process.
Moreover, the next lemma will show that it is a finite variation continuous process.

Setting $Y:=X$ leads to the definition of the \emph{quadratic variation} $[X,X]$,
which we will sometimes abbreviate to $[X]$.
It is clear from the definition \eqref{eq:covariation}
that $[X]$ is an increasing and, therefore, finite variation continuous process.
The following lemma shows that $[X,Y]$ is a finite variation continuous process
for any continuous martingales $X$ and $Y$
(and, as Lemma~\ref{lem:invariant-covariation} below will show,
even for any continuous semimartingales $X$ and $Y$).

\begin{lemma}\label{lem:finite-variation}
  For any continuous martingales $X$ and $Y$,
  \begin{equation*}
    [X,Y]
    =
    \frac12
    \left(
      [X+Y] - [X] - [Y]
    \right)
    \text{\quad q.a.},
  \end{equation*}
  and $[X,Y]$ is a finite variation continuous process.
\end{lemma}

\begin{proof}
  The identity
  $ab=\frac12((a+b)^2-a^2-b^2)$ implies
  \[
    [X,Y]^n
    =
    \frac12
    \left(
      [X+Y]^n - [X]^n - [Y]^n
    \right)
  \]
  for each $n=1,2,\ldots$,
  and it remains to pass to a ucqa limit as $n\to\infty$.
\end{proof}

Let us now extend the notions of covariation and quadratic variation to continuous semimartingales.
Again our previous definition of $[X,Y]$ for continuous martingales carries over
to the case of continuous semimartingales $X$ and $Y$ verbatim
(using a sequence of partitions that is fine for all components of the standard decompositions of $X$ and $Y$),
and it is clear that Lemma~\ref{lem:covariation} holds for any continuous semimartingales.
Quadratic variation can still be defined as $[X]:=[X,X]$.

As in measure-theoretic probability, the covariation between two continuous semimartingales
only depends on their martingale parts.

\begin{lemma}\label{lem:invariant-covariation}
  If $X$ and $X'$ are two continuous semimartingales
  with standard decompositions $X=Y+A$ and $X'=Y'+A'$,
  then $[X,X']=[Y,Y']$ q.a.
\end{lemma}

\begin{proof}
  By the definition \eqref{eq:covariation} of covariation,
  \begin{align}
    [X,X']^n_t(\omega)
    &=
    \sum_{k=1}^{\infty}
    \left(
      X_{T^n_k\wedge t}
      -
      X_{T^n_{k-1}\wedge t}
    \right)
    \left(
      X'_{T^n_k\wedge t}
      -
      X'_{T^n_{k-1}\wedge t}
    \right)\notag\\
    &=
    \sum_{k=1}^{\infty}
    \left(
      Y_{T^n_k\wedge t}
      -
      Y_{T^n_{k-1}\wedge t}
    \right)
    \left(
      Y'_{T^n_k\wedge t}
      -
      Y'_{T^n_{k-1}\wedge t}
    \right)\label{eq:first}\\
    &\quad{}+
    \sum_{k=1}^{\infty}
    \left(
      Y_{T^n_k\wedge t}
      -
      Y_{T^n_{k-1}\wedge t}
    \right)
    \left(
      A'_{T^n_k\wedge t}
      -
      A'_{T^n_{k-1}\wedge t}
    \right)\label{eq:second}\\
    &\quad{}+
    \sum_{k=1}^{\infty}
    \left(
      A_{T^n_k\wedge t}
      -
      A_{T^n_{k-1}\wedge t}
    \right)
    \left(
      Y'_{T^n_k\wedge t}
      -
      Y'_{T^n_{k-1}\wedge t}
    \right)\label{eq:third}\\
    &\quad{}+
    \sum_{k=1}^{\infty}
    \left(
      A_{T^n_k\wedge t}
      -
      A_{T^n_{k-1}\wedge t}
    \right)
    \left(
      A'_{T^n_k\wedge t}
      -
      A'_{T^n_{k-1}\wedge t}
    \right).\label{eq:fourth}
  \end{align}
  Since the first addend \eqref{eq:first} in the last sum is $[Y,Y']^n_t$,
  we are required to show that the other three addends, \eqref{eq:second}--\eqref{eq:fourth}, converge to zero as $n\to\infty$.
  The same argument works for all three addends;
  e.g., \eqref{eq:second} tends to zero because
  \[
    \left|
      \sum_{k=1}^{\infty}
      \left(
        Y_{T^n_k\wedge t}
        -
        Y_{T^n_{k-1}\wedge t}
      \right)
      \left(
        A'_{T^n_k\wedge t}
        -
        A'_{T^n_{k-1}\wedge t}
      \right)
    \right|
    \le
    2^{-n}
    O(1)
    \to
    0
    \quad
    (n\to\infty),
  \]
  where we have used the fineness of the sequence of partitions and the finite variation of $A'$.
\end{proof}

\section{It\^o formula}

We start from stating the It\^o formula for continuous semimartingales.

\begin{theorem}\label{thm:Ito}
  Let $F:\bbbr\to\bbbr$ be a function of class $C^2$
  and $X$ be a continuous semimartingale.
  Then
  \begin{equation*}
    F(X_t)
    =
    F(X_0)
    +
    \int_0^t
    F'(X_s)
    \dd X_s
    +
    \frac12
    \int_0^t
    F''(X_s)
    \dd[X]_s
    \text{\quad q.a.}
  \end{equation*}
\end{theorem}

\noindent
The last integral
$
  \int_0^t
  F''(X_s)
  \dd[X]_s
$
can be understood in the Lebesgue--Stiltjes sense.

\begin{proof}
  By Taylor's formula,
  \begin{multline*}
    F(X_{T^n_k}) - F(X_{T^n_{k-1}})
    =
    F'(X_{T^n_{k-1}})
    \Bigl(
      X_{T^n_k} - X_{T^n_{k-1}}
    \Bigr)\\
    +
    \frac12
    F''(\xi_k)
    \Bigl(
      X_{T^n_k} - X_{T^n_{k-1}}
    \Bigr)^2,
  \end{multline*}
  where $\xi_k\in[X_{T^n_{k-1}},X_{T^n_{k}}]$
  (and $[a,b]$ is understood to be $[b,a]$ when $a>b$).
  It remains to sum this equality over $k=1,\ldots,K$,
  where $K$ is the largest $k$ such that $T^n_k\le t$,
  and to pass to the limit as $n\to\infty$.
\end{proof}

The next result is a vector version of Theorem~\ref{thm:Ito} and is proved in a similar way.
By a \emph{vector continuous semimartingale} we mean a finite sequence
$X=(X^1,\ldots,X^d)$ of continuous semimartingales considered as a function mapping $(t,\omega)\in[0,\infty)\times\Omega$
to the vector $X_t(\omega)=(X^1_t(\omega),\ldots,X^d_t(\omega))$.

\begin{theorem}\label{thm:Ito-MD}
  Let $F:\bbbr^d\to\bbbr$ be a function of class $C^2$
  and $X=(X^1,\ldots,X^d)$ be a vector continuous semimartingale.
  Then
  \begin{multline}\label{eq:Ito-MD}
    F(X_t)
    =
    F(X_0)
    +
    \sum_{i=1}^d
    \int_0^t
    \frac{\partial F}{\partial x_i}(X_s)
    \dd X^i_s\\
    +
    \frac12
    \sum_{i=1}^d
    \sum_{j=1}^d
    \int_0^t
    \frac{\partial^2 F}{\partial x_i \partial x_j}(X_s)
    \dd[X^i,X^j]_s
    \text{\quad q.a.}
  \end{multline}
\end{theorem}

\begin{remark}
  The requirement that $F$ be twice continuously differentiable can be relaxed for the components
  for which $X^i$ has a special form, such as $X^i_t=t$ for all $t$.
  This, however, will not be needed in this paper.
\end{remark}

\section{Dol\'eans exponential and logarithm}

The following theorem introduces a game-theoretic analogue of the Dol\'eans exponential.

\begin{theorem}\label{thm:Doleans-exp}
  If $X$ is a continuous martingale,
  $\EEE(X):=\exp(X-[X]/2)$ is a continuous martingale as well.
\end{theorem}

\begin{proof}
  A standard trick
  (cf.\ \cite{Revuz/Yor:1999}, Proposition IV.3.4)
  is to apply the It\^o formula \eqref{eq:Ito-MD}
  to the function $F(x,y)=\exp(x-y/2)$ and vector continuous semimartingale $(X,Y)=(X,[X])$.
  Since $[X,[X]]=0$ (cf.\ the proof of Lemma~\ref{lem:invariant-covariation}), $[[X],[X]]=0$, and
  \begin{equation*}
    \frac{\partial F}{\partial y}
    +
    \frac12 
    \frac{\partial^2 F}{\partial x^2}
    =
    0,
  \end{equation*}
  we have, by the It\^o formula,
  \begin{equation}\label{eq:SDE-1}
    F(X_t,[X]_t)
    =
    F(X_0,0)
    +
    \int_0^t
    \frac{\partial F}{\partial x}(X_s,[X]_s)
    \dd X_s
    \text{\quad q.a.};
  \end{equation}
  therefore, $F(X_t,[X]_t)$ is a continuous martingale
  (by Lemma~\ref{lem:stochastic-integral}).
\end{proof}

\begin{remark}
  Since $\partial F/\partial x = F$,
  \eqref{eq:SDE-1} can be rewritten as the stochastic differential equation
  \begin{equation}\label{eq:SDE-2}
    Y_t
    =
    Y_0
    +
    \int_0^t
    Y_s
    \dd X_s
  \end{equation}
  for $Y:=\exp(X-[X]/2)$;
  the Dol\'eans exponential is its solution.
\end{remark}

Later in this paper we will be given a positive continuous martingale $I$
and will be interested in a continuous martingale $L$ such that $I$ is the Dol\'eans exponential for $L$;
therefore, we are are interested in an inverse operation to taking the Dol\'eans exponential.
(See, e.g., \cite{Jacod/Shiryaev:2003-local}, Section~II.8a, for a measure-theoretic exposition.)

The \emph{Dol\'eans logarithm} $X$ of a positive continuous martingale $Y$ can be defined in two different ways:
by the It\^o integral
\begin{equation}\label{eq:definition-1}
  X_t
  :=
  \ln Y_0
  +
  \int_0^t
  \frac{\dd Y_s}{Y_s}
\end{equation}
and by the more explicit formula
\begin{equation}\label{eq:definition-2}
  X_t
  :=
  \ln Y_t
  +
  \frac12
  [\ln Y]_t.
\end{equation}
The two definitions are equivalent, but we will only check that \eqref{eq:definition-2} implies \eqref{eq:definition-1}
(and so \eqref{eq:definition-2} can be taken as the main definition).
Applying the It\^o formula to the function
\[
  F(y_1,y_2)
  :=
  \ln y_1
  +
  \frac12
  y_2
\]
and the continuous semimartingales $Y$ and $[\ln Y]$,
we obtain the first definition \eqref{eq:definition-1} from the second definition \eqref{eq:definition-2}:
\begin{align*}
  X_t
  &=
  \ln Y_t
  +
  \frac12
  [\ln Y]_t
  =
  F(Y_t,[\ln Y]_t)\\
  &=
  F(Y_0,0)
  +
  \int_0^t
  \frac{\partial F}{\partial y_1}(Y_s,[\ln Y]_s)
  \dd Y_s
  +
  \int_0^t
  \frac{\partial F}{\partial y_2}(Y_s,[\ln Y]_s)
  \dd[\ln Y]_s\\
  &\quad{}+
  \frac12
  \int_0^t
  \frac{\partial^2 F}{\partial y_1^2}(Y_s,[\ln Y]_s)
  \dd[Y]_s\\
  &=
  \ln Y_0
  +
  \int_0^t
  \frac{\dd Y_s}{Y_s}
  +
  \frac12
  \int_0^t
  \dd[\ln Y]_s
  -
  \frac12
  \int_0^t
  \frac{\dd[Y]_s}{Y_s^2}
  =
  \ln Y_0
  +
  \int_0^t
  \frac{\dd Y_s}{Y_s}
  \text{\quad q.a.}
\end{align*}
(The last equality follows from $[\ln Y]_t=\int_0^t\dd[Y]_s/Y_s^2$,
which is easy to check and will be generalized in \eqref{eq:Sigma} below.)
The first definition \eqref{eq:definition-1} shows that the Dol\'eans logarithm of a positive continuous martingale
is a continuous martingale.
The following theorem summarizes our discussion so far in this section
adding a couple of trivial observations.

\begin{theorem}
  If $Y$ is a positive continuous martingale,
  $\LLL(Y):=\ln Y+[\ln Y]/2$ is a continuous martingale.
  For any continuous martingale $X$, $\LLL(\EEE(X))=X$ q.a.
  For any positive continuous martingale $Y$, $\EEE(\LLL(Y))=Y$ q.a.
\end{theorem}

\begin{remark}
  Informally, \eqref{eq:SDE-2} and \eqref{eq:definition-1} can be rewritten as
  $\dd Y_t=Y_t\dd X_t$ and $\dd X_t=\dd Y_t/Y_t$,
  respectively;
  in this form their similarity is more obvious.
\end{remark}

\section{Probability-free Dubins--Schwarz theo\-rem for continuous mar\-tin\-gales}
\label{sec:Dubins-Schwarz}

In this and next sections we will make two essential steps.
First, in this section we will define a general notion of game-theoretic upper probability.
So far the only probability-type property that we have used
was that of ``quasi-always,''
which is closely connected with events of upper probability zero
(as explained later in this section).
Second, in the next section we will start discussing using a num\'eraire different from cash
(which has been implicitly used so far).

It is shown in \cite{\XXVIII} that, roughly, a continuous price path can be transformed into a Brownian motion
by replacing physical time with quadratic variation.
This time we apply this idea in a way that is closer to the classical Dubins--Schwarz result,
replacing a continuous price path by a continuous martingale
and using Theorem~\ref{thm:conservative-super}.

The initial value $X_0$ of a nonnegative supermartingale $X$ is always a constant.
Given a functional $F:\Omega\to[0,\infty)$, we define its \emph{upper expectation} as
\begin{equation*}
  \UpExpect(F)
  :=
  \inf
  \bigl\{
    X_0
    \bigm|
    \forall\omega\in\Omega:
    \liminf_{t\to\infty}
    X_t(\omega)
    \ge
    F(\omega)
  \bigr\},
\end{equation*}
$X$ ranging over the nonnegative supermartingales.
The \emph{upper probability} $\UpProb(E)$ of a set $E\subseteq\Omega$ is defined as $\UpExpect(\III_E)$,
where $\III_E$ is the indicator function of $E$.
A property of $\omega\in\Omega$ holds \emph{almost surely} (a.s.)\ if its complement
has upper probability zero.
These are standard definitions using cash as num\'eraire
(in the terminology of the next section).
As we mentioned earlier,
the projection onto $\Omega$ of the complement of a property of $t$ and $\omega$ that holds q.a.\
always has upper probability zero.

A \emph{time transformation} is defined to be a continuous increasing
(not necessarily strictly increasing)
function $f:[0,\infty)\to[0,\infty)$ satisfying $f(0)=0$.
A nonnegative functional $F:\Omega\to[0,\infty]$ is \emph{time-superinvariant}
if, for each $\omega\in\Omega$ and each time transformation $f$,
\begin{equation*}
  F(\omega\circ f)
  \le
  F(\omega).
\end{equation*}

\begin{theorem}\label{thm:Dubins-Schwarz}
  Let $F:C[0,\infty)\to[0,\infty]$ be a time-superinvariant $\FFF$-measurable functional,
  and let $X$ be a continuous martingale.
  Then
  \begin{equation}\label{eq:Dubins-Schwarz}
    \UpExpect(F(X))
    \le
    \int F \dd W_{X_0},
  \end{equation}
  where $W_{X_0}$ is Brownian motion starting from $X_0$.
\end{theorem}

In \eqref{eq:Dubins-Schwarz} we set $F(X):=\infty$ when $X\notin C[0,\infty)$.

\begin{proof}
  Combine the part $\le$ of Theorem~6.3 in \cite[technical report]{\XXVIII} with our Theorem~\ref{thm:conservative-super}.
  The former theorem now simplifies since the initial value of a continuous martingale is a constant;
  even though that theorem assumes $F:C[0,\infty)\to[0,\infty)$, its proof also works for $[0,\infty]$ in place of $[0,\infty)$.
  (Notice that the simple part $\ge$ of that theorem
  is not applicable anymore
  since the range of $X$ can contain far from all continuous paths starting at~$X_0$.)
\end{proof}

\begin{corollary}\label{cor:decomposition}
  The decomposition of a continuous semimartingale $X$ into the sum $X=Y+A$
  of a continuous martingale $Y$ and a finite variation continuous process $A$
  is unique (q.a.).
\end{corollary}

\noindent
The detailed statement of the corollary is:
if $X=Y+A$ and $X=Y'+A'$ are two such decompositions,
then, quasi-always, $Y|_{[0,t]}=Y'|_{[0,t]}$ and $A|_{[0,t]}=A'|_{[0,t]}$.

\begin{proof}
  Let $X=Y+A=Y'+A'$ be two such decompositions;
  then $Y-Y'=A'-A$, and so we have a continuous martingale
  which is simultaneously a finite variation continuous process.
  Define $F:C[0,\infty)\to[0,\infty]$ to be the indicator functional of a function in $C[0,\infty)$
  having a finite variation over some interval $[0,t]$, $t\in(0,\infty)$,
  while not being constant over that interval.
  By Theorem~\ref{thm:Dubins-Schwarz},
  there exists a nonnegative supermartingale $Z$ with $Z_0=1$
  that tends to $\infty$ on $\omega$
  such that $Y-Y'$ has a finite variation over some $[0,t]$
  without being constant over that interval.
  This nonnegative supermartingale $Z$ will tend to $\infty$
  on any continuation of $\omega|_{[0,t]}$ such that 
  $(Y(\omega)-Y'(\omega))|_{[0,t]}=(A'(\omega)-A(\omega))|_{[0,t]}\ne0$.
  This means that already $Z_t(\omega)=\infty$,
  since we can extend $\omega|_{[0,t]}$ by a constant.
\end{proof}

\Subsection{Comparison between game-theoretic and measure-theoretic probability}

We again consider the measure-theoretic setting introduced at the end of Section~\ref{sec:martingales}.
Let us now check that $\UpExpect(F)\ge\Expect(F)$ for each $\FFF$-measurable nonnegative functional $F:\Omega\to[0,\infty)$,
which will imply that $\UpProb(E)\ge\Prob(E)$ for each $\FFF$-measurable $E\subseteq\Omega$.
(In this sense our definition of upper probability $\UpProb$
is not too permissive, unlike the definition ignoring measurability in \cite{\XLIII}.)
It suffices to prove that, for any nonnegative measure-theoretic supermartingale $X$ with a constant $X_0$,
\[
  \liminf_{t\to\infty} X_t \ge F
  \Longrightarrow
  X_0 \ge \Expect(F).
\]
This can be done using the Fatou lemma:
assuming the antecedent,
\[
  \Expect(F)
  \le
  \Expect
  \left(
    \liminf_{t\to\infty} X_t
  \right)
  \le
  \liminf_{t\to\infty} \Expect(X_t)
  \le
  \liminf_{t\to\infty} \Expect(X_0)
  =
  X_0.
\]

\section{General num\'eraires}
\label{sec:numeraire}

In this section we fix a positive continuous martingale $I:[0,\infty)\times\Omega\to(0,\infty)$
and use it as our num\'eraire for measuring capital at time $t$.
The results in the previous sections can be regarded as a special case
corresponding to $I:=1$ (intuitively, using cash as the num\'eraire).
Generalization (``relativization,'' to use an expression from the theory of computability \cite{Rogers:1967-local}, Section 9.2)
of results with cash as num\'eraire to general num\'eraires
is easy when only one num\'eraire is involved,
but in the next section, devoted to a probability-free version of Girsanov's theorem,
we will see a nontrivial result involving two different num\'eraires.

We start from generalizing the definitions of Section~\ref{sec:martingales}.
We extend our market by adding another continuous path $\bar S_0$
which we will interpret as the unit, ``cash,'' in which the prices $S_1,\ldots,S_{J^*}$ are measured.
Our old unit $S_0:=1$ (which we did not need to mention explicitly) is expressed as $\bar S_0$ in terms of the new unit.
Now we have $J^*+1$ traded securities $\bar S_0,\bar S_1,\ldots,\bar S_{J^*}$,
where $\bar S_j:=\bar S_0 S_j$ for $j=1,\ldots,J^*$.
A \emph{simple trading strategy} $\bar G$ now consists of stopping times $\tau_n$, as before,
and also bounded $\tau_n$-measurable $\bbbr^{J^*+1}$-valued random vectors,
which we will denote $\bar h_n=(H_n,h_n)$, where $H_n$ are random variables and $h_n$ are $\bbbr^{J^*}$-valued random vectors.
In the new picture no borrowing or lending are allowed
(they should be done implicitly via investing in the available $J^*+1$ securities).
The capital available at time $t$ is
\begin{equation}\label{eq:simple-capital-2}
  \bar\K_t
  :=
  \bar h_n \cdot \bar \omega^*(t),
\end{equation}
where $n$ is such that $\tau_n\le t\le\tau_{n+1}$ and we use the notation
\[
  \bar \omega^*
  :=
  (\bar S_0,\tilde \omega^*)
  :=
  (\bar S_0,\bar S_1,\ldots,\bar S_{J^*}).
\]
The simple trading strategy $\bar G$ is required to be \emph{self-financing} in the sense that,
for any $n=2,3,\ldots$,
\begin{equation}\label{eq:self-financing}
  \bar h_{n-1} \cdot \bar \omega^*(\tau_n)
  =
  \bar h_{n} \cdot \bar \omega^*(\tau_n)
\end{equation}
(this property implies that the expression \eqref{eq:simple-capital-2} is well defined
when $t=\tau_i$ for some $i$).
Now the strategy determines its (random) \emph{initial capital}
\begin{equation}\label{eq:initial}
  \bar\K_{\tau_1}
  =
  \bar h_1\cdot\bar\omega^*(\tau_1).
\end{equation}

In the generalized picture we have a symmetry among the $J^*+1$ traded securities
with price paths $(\bar S_0,\ldots,\bar S_{J^*})$;
we chose to use $\bar S_0$ as our num\'eraire but could have chosen any other security
with a positive price path
(we can always restrict the market in a certain way, such as making a basic price path positive).
Let us check that the generalized picture gives the notion of capital,
which we will denote $\bar\K$,
that agrees with our original picture.
By the condition \eqref{eq:self-financing} of being self-financing,
the number $H_n$ of units of cash chosen at time $\tau_n$ should satisfy
\begin{equation}\label{eq:later}
  \bar\K_{\tau_n}
  =
  \bar h_n\cdot\bar\omega^*(\tau_n)
  =
  H_n \bar S_0(\tau_n) + h_n\cdot\tilde\omega^*(\tau_n),
\end{equation}
which gives
\[
  H_n
  =
  \frac{\bar\K_{\tau_n} - h_n\cdot\tilde\omega^*(\tau_n)}{\bar S_0(\tau_n)}.
\]
Notice that this is also true for $n=1$,
in which case we should use \eqref{eq:initial} rather than \eqref{eq:later}.
Therefore, the capital at time $\tau_{n+1}$ becomes
\begin{align*}
  \bar\K_{\tau_{n+1}}
  &=
  H_n \bar S_0(\tau_{n+1}) + h_n\cdot\tilde\omega^*(\tau_{n+1})\\
  &=
  \frac{\bar S_0(\tau_{n+1})}{\bar S_0(\tau_n)}
  \left(
    \bar\K_{\tau_n} - h_n\cdot\tilde\omega^*(\tau_n)
  \right)
  +
  h_n\cdot\tilde\omega^*(\tau_{n+1}),
\end{align*}
which in the units of $\bar S_0$, $\K_t:=\bar\K_t/\bar S_0(t)$, becomes
\[
  \K_{\tau_{n+1}}
  =
  \K_{\tau_n} - h_n\cdot\omega^*(\tau_n)
  +
  h_n\cdot\omega^*(\tau_{n+1})
  =
  \K_{\tau_n}
  +
  h_n\cdot(\omega^*(\tau_{n+1})-\omega^*(\tau_n)).
\]
Since this is also true for any $t\in[\tau_n,\tau_{n+1}]$ in place of $\tau_{n+1}$,
we obtain \eqref{eq:simple-capital-1} for all $t\ge\tau_1$
(where $c$ is the initial capital \eqref{eq:initial} expressed in units of $\bar S_0$).
For $t\le\tau_1$,
we stipulate that the strategy has an initial endowment of $c$ units of $\bar S_0$
(where $c$ is a given constant),
which makes \eqref{eq:simple-capital-1} true for all $t\in[0,\infty)$.

Let $I$ be any of the securities $S_j>0$ among $S_0=1,S_1,\ldots,S_{J^*}$
in the original market ($S_j$ being positive is our restriction on the market);
in view of Theorems~\ref{thm:conservative-super}--\ref{thm:conservative}
we will later allow $I$ to be any positive continuous martingale.
We have just seen that for any simple capital process $X$
(in our original picture with cash $S_0=1$ as the num\'eraire),
the process $X/I=X\bar S_0/\bar S_j$ will be a simple capital process
in the picture with $I$ as the num\'eraire.
A simple argument based on transfinite induction shows that this statement can be extended to continuous martingales:
for any continuous martingale $X$,
the process $X/I$ will be a continuous martingale in the picture with $I$ as the num\'eraire;
the inductive step is based on the identity
\begin{equation}\label{eq:identity}
  \lim_{k\to\infty}
  \frac{X^k_t(\omega)}{I_t(\omega)}
  =
  \frac{\lim_{k\to\infty}X^k_t(\omega)}{I_t(\omega)}.
\end{equation}
It is also true that, for any nonnegative supermartingale $X$,
the process $X/I$ will be a nonnegative supermartingale in the picture with $I$ as the num\'eraire;
the inductive step is now based on the identity \eqref{eq:identity} with $\liminf$ in place of $\lim$.

We call processes of the type $X_t(\omega) / I_t(\omega)$, 
where $X$ is a nonnegative supermartingale (resp., a continuous martingale)
\emph{nonnegative $I$-supermartingales} (resp., \emph{continuous $I$-martingales});
they are completely analogous to nonnegative supermartingales (resp., continuous martingales)
but use $I$ rather than cash as the num\'eraire.
The processes $X$ of the form $Y+A$,
where $X$ is a continuous $I$-martingale and $A$ is a finite variation continuous process,
are \emph{continuous $I$-semimartingales}.
Whereas the first two notions very much depend on the choice of $I$,
the Girsanov theorem in the next section will show that the notion of a continuous $I$-semimartingale is invariant.

\Subsection{It\^o integration and its applications for a general num\'eraire}

Theorem~\ref{thm:Ito-integral} remains true if $X$ is a continuous $I$-martingale
rather than a continuous martingale because it involves the notion of ``quasi-always,''
which is defined in terms of becoming infinitely rich infinitely quickly
and so does not depend on the num\'eraire.
It is obvious that Lemma~\ref{lem:stochastic-integral} remains true
if the two entries of ``continuous martingale'' are replaced by ``continuous $I$-martingale.''
In Corollary~\ref{cor:invariance} we can allow $X$ to be a continuous $I$-martingale
(and in its interpretation we can allow definitions of sequences of partitions to depend on $I$).
Finally, the definition of $H\cdot X$ carries over to continuous $I$-semimartingales.
Notice that $H\cdot X$ does not depend on $I$.

We have the following corollary of Lemma~\ref{lem:covariation}.

\begin{corollary}
  Every pair of continuous $I$-martingales $X,Y$ possesses covariation $[X,Y]$ q.a.,
  which satisfies the integration by part formula \eqref{eq:by-parts}.
\end{corollary}

\noindent
Lemma~\ref{lem:invariant-covariation} continues to hold for continuous $I$-semimartingales.
Theorems~\ref{thm:Ito} and~\ref{thm:Ito-MD} carry over to continuous $I$-semimartingales verbatim.
The definitions of the Dol\'eans exponential and logarithm and their properties
carry over verbatim to the case of continuous $I$-martingales.
For example, Theorem~\ref{thm:Doleans-exp} implies:

\begin{corollary}
  If $X$ is a continuous $I$-martingale,
  $\exp(X-[X]/2)$ is a continuous $I$-martingale.
\end{corollary}

\Subsection{Dubins--Schwarz theorem for a general num\'eraire}

Let $I$ be a positive continuous martingale.
We define \emph{upper $I$-expectation} by
\begin{equation}\label{eq:upper-expectation-I}
  \UpExpectI(F)
  :=
  \inf
  \bigl\{
    X_0
    \bigm|
    \forall\omega\in\Omega:
    \liminf_{t\to\infty}
    X_t(\omega) / I_t(\omega)
    \ge
    F(\omega)
  \bigr\},
\end{equation}
$X$ ranging over the nonnegative supermartingales and $F$ over nonnegative functionals,
and specialize it to \emph{upper $I$-probability} as $\UpProbI(E):=\UpExpectI(\III_E)$.
The definition~\eqref{eq:upper-expectation-I} can be rewritten as
\begin{equation*}
  \UpExpectI(F)
  =
  \inf
  \bigl\{
    X_0
    \bigm|
    \forall\omega\in\Omega:
    \liminf_{t\to\infty}
    X_t(\omega)
    \ge
    F(\omega)
  \bigr\},
\end{equation*}
$X$ ranging over the nonnegative $I$-supermartingales.

The results of Section~\ref{sec:Dubins-Schwarz} carry over
to the case of a general positive continuous martingale $I$ as num\'eraire.
In particular, we have the following version of Theorem~\ref{thm:Dubins-Schwarz}.

\begin{corollary}\label{cor:Dubins-Schwarz}
  Let $F:C[0,\infty)\to[0,\infty]$ be a time-superinvariant $\FFF$-measurable functional,
  $I$ be a positive continuous martingale,
  and $X$ be a continuous $I$-martingale.
  Then
  \[
    \UpExpectI(F(X))
    \le
    \int F \dd W_{X_0}.
  \]
\end{corollary}

\section{Girsanov theorem}

Now we state a probability-free Girsanov theorem,
the main result of this section.
It shows that the notion of a continuous semimartingale does not depend on the num\'eraire
and gives the explicit decomposition (unique by Corollary~\ref{cor:decomposition})
into a continuous martingale and a finite variation continuous process
for a continuous martingale in a new num\'eraire.

\begin{theorem}\label{thm:Girsanov}
  Let $M$ be a continuous martingale and $I$ be a positive continuous martingale.
  The process
  \begin{equation}\label{eq:Girsanov}
    M_t - \int_0^t \frac{\dd[I,M]_s}{I_s}
  \end{equation}
  (where the integral is Lebesgue--Stiltjes; cf.\ Lemma~\ref{lem:finite-variation})
  is a continuous $I$-martingale.
\end{theorem}

\begin{proof}
  Our proof will be standard
  (see, e.g., Protter \cite{Protter:2005-local}, the proof of Theorem~III.39).

  Remember that, by the integration by parts formula (see Lemma~\ref{lem:covariation}),
  \[
    I_t M_t
    =
    (I\cdot M)_t
    +
    (M\cdot I)_t
    +
    [I,M]_t
    \quad
    \text{q.a.}
  \]
  Since $I\cdot M$ and $M\cdot I$ are continuous martingales,
  \[
    I_t M_t
    -
    [I,M]_t
  \]
  is also a continuous martingale,
  and so
  \[
    M_t
    -
    \frac{1}{I_t}
    [I,M]_t
  \]
  is a continuous $I$-martingale.
  The integration by parts formula 
  (Lemma~\ref{lem:covariation},
  which is also applicable to continuous semimartingales, with the same proof)
  allows us to transform the subtrahend (the product of continuous $I$-semimartingales) as
  \begin{equation}\label{eq:subtrahend}
    \frac{1}{I_t}
    [I,M]_t
    =
    \int_0^t \frac{\dd[I,M]_s}{I_s}
    +
    \int_0^t [I,M]_s \dd\frac{1}{I_s}
    +
    \left[
      \frac{1}{I},
      [I,M]
    \right]_t
    \quad
    \text{q.a.}
  \end{equation}
  The second addend on the right-hand side of \eqref{eq:subtrahend} is a continuous $I$-martingale
  since $1/I_t$ is,
  and the third addend is a continuous $I$-martingale
  since it is zero
  q.a.\ (see the argument at the end of the proof of Lemma~\ref{lem:invariant-covariation});
  therefore, \eqref{eq:Girsanov} is also a continuous $I$-martingale.
\end{proof}

The notion of Dol\'eans logarithm allows us to simplify the statement of Theorem~\ref{thm:Girsanov} as follows.

\begin{corollary}\label{cor:Girsanov}
  Let $M$ be a continuous martingale, $I$ be a positive continuous martingale, and $L$ be the Dol\'eans logarithm of $I$.
  The process
  \begin{equation*}
    M_t - [L,M]_t
  \end{equation*}
  is a continuous $I$-martingale.
\end{corollary}

\begin{proof}
  It suffices to prove
  \begin{equation*}
    \int_0^t \frac{\dd[I,M]_s}{I_s}
    =
    [L,M]_t
    \quad
    \text{q.a.};
  \end{equation*}
  remember that the integral on the left-hand side is the usual Lebesgue--Stiltjes integral.
  For sufficiently fine sequences of partitions, we have:
  \begin{align}
    \int_0^t \frac{\dd[I,M]_s}{I_s}
    &=
    \lim_{n\to\infty}
    \sum_{k=1}^{\infty}
    \frac
    {
      [I,M]_{T^n_k\wedge t}
      -
      [I,M]_{T^n_{k-1}\wedge t}
    }
    {I_{T^n_{k-1}\wedge t}}
    \label{eq:fine-1}\\
    &=
    \lim_{n\to\infty}
    \sum_{k=1}^{\infty}
    \frac
    {
      \left(
        I_{T^n_k\wedge t}
        -
        I_{T^n_{k-1}\wedge t}
      \right)
      \left(
        M_{T^n_k\wedge t}
        -
        M_{T^n_{k-1}\wedge t}
      \right)
    }
    {I_{T^n_{k-1}\wedge t}}
    \label{eq:fine-2}\\
    &=
    \lim_{n\to\infty}
    \sum_{k=1}^{\infty}
    \frac
    {
      I_{T^n_k\wedge t}
      -
      I_{T^n_{k-1}\wedge t}
    }
    {I_{T^n_{k-1}\wedge t}}
    \left(
      M_{T^n_k\wedge t}
      -
      M_{T^n_{k-1}\wedge t}
    \right)
    \label{eq:fine-3}\\
    &=
    \lim_{n\to\infty}
    \sum_{k=1}^{\infty}
    \left(
      L_{T^n_k\wedge t}
      -
      L_{T^n_{k-1}\wedge t}
    \right)
    \left(
      M_{T^n_k\wedge t}
      -
      M_{T^n_{k-1}\wedge t}
    \right)
    \label{eq:fine-4}\\
    &=
    [L,M]_t
    \quad
    \text{q.a.}
    \label{eq:final}
  \end{align}
  (The transition from \eqref{eq:fine-1} to \eqref{eq:fine-2} requires refining the partitions:
  the expression after $\lim_{n\to\infty}$ in \eqref{eq:fine-1} for a given large value of $n$
  is approximately equal to the expression after $\lim_{n\to\infty}$ in \eqref{eq:fine-2} for much larger values of $n$.
  A similar remark can be made about the transition from \eqref{eq:fine-4} to \eqref{eq:fine-3}.)
\end{proof}

\section{Applications to the equity premium and CAPM}

In this section we rederive, generalize, and strengthen various results in \cite{\XLIV}.
Let us fix two positive continuous martingales, $S$ and $I$.
We interpret $S$ as a stock and $I$ as an index (something like S\&P500),
so that, e.g., $S$ can be one of the traded $S_j$ and $I$ can be their weighted average
(under the restriction that all basic price paths are positive).

For a positive continuous martingale $X$ (interpreted as the price of a financial security),
we define its \emph{cumulative relative growth} as
\[
  \Mu^X_t(\omega)
  :=
  \int_0^t
  \frac{\dd X_s}{X_s}
  =
  \Lambda^X_t
  -
  \ln X_0,
\]
where $\Lambda^X$ stands for the Dol\'eans logarithm of $X$.
(In this section we mainly follow the terminology of \cite{\XLIV}.)
The \emph{relative covariation} of positive continuous martingales $X$ and $Y$ is the Lebesgue--Stiltjes integral
\begin{equation}\label{eq:Sigma}
  \Sigma^{X,Y}_t(\omega)
  :=
  \int_0^t
  \frac{\dd[X,Y]_s}{X_s Y_s}
  =
  [\Lambda^X,\Lambda^Y]_t
  \quad
  \text{q.a.}
\end{equation}
(cf.\ Lemma~\ref{lem:finite-variation};
the expression in terms of the Dol\'eans logarithms
can be derived similarly to \eqref{eq:fine-1}--\eqref{eq:final}).
The following corollary of our probability-free Girsanov theorem
is a key result implying all others in this section.

\begin{theorem}[\cite{\XLIV}, Theorem~8.3]
  \label{thm:key}
  The process $\Mu^S_t - \Sigma^{S,I}_t$ is a continuous $I$-martingale.
\end{theorem}

\begin{proof}
  Applying Corollary~\ref{cor:Girsanov}
  to the cumulative relative growth $M:=\Mu^S$,
  we obtain that
  \begin{equation*}
    \Mu^S_t - [\Mu^S,\Lambda^I]_t
    =
    \Mu^S_t - \Sigma^{S,I}_t
  \end{equation*}
  is a continuous $I$-martingale.
\end{proof}

Since Theorem~\ref{thm:key} is applicable to any pair $(S,I)$ of positive continuous martingales,
we can replace $S$ by $I$ obtaining the following equity premium result,
in which $\Sigma^I:=\Sigma^{I,I}$ stands for the \emph{relative quadratic variation} of $I$.

\begin{corollary}[\cite{\XLIV}, Corollary~8.2]
  \label{cor:equity-premium}
  The process $\Mu^I_t - \Sigma^{I}_t$ is a continuous $I$-martingale.
\end{corollary}

We can apply various results of this paper to the continuous $I$-martingale of Theorem~\ref{thm:key}
(and then specialize them to the continuous $I$-martingale of Corollary~\ref{cor:equity-premium}).

  \begin{corollary}[\cite{\XLIV}, Lemma~8.5]
    \label{cor:exp}
    For each $\epsilon\in\bbbr$, the process
    \begin{equation*}
      \exp
      \left(
        \epsilon (\Mu^S_t - \Sigma^{S,I}_t)
        -
        \frac{\epsilon^2}{2} \Sigma^{S}_t
      \right)
    \end{equation*}
    is a continuous $I$-martingale.
  \end{corollary}

  \begin{proof}
    Combine Theorems~\ref{thm:key} and~\ref{thm:Doleans-exp}.
  \end{proof}

  Corollary~\ref{cor:exp} played an important role in \cite{\XLIV},
  allowing us to derive analogues of the following two corollaries.
  It does not play any special role in our current exposition
  (although it might become important again
  if we allow jumps in the basic price paths).

  The following corollary strengthens Corollary~8.6 of \cite{\XLIV}
  optimizing the inequality in it
  (see Figure~1 in \cite{\XLIV} for an illustration of the difference between the old and new inequalities).

\begin{corollary}\label{cor:CAPM-CLT}
  If $\delta>0$ and $\tau_T:=\inf\{t\st \Sigma^{S}_t\ge T\}$ for some constant $T>0$,
  \begin{equation}\label{eq:inequality}
    \UpProbI
    \left\{
      \left|
        \Mu^S_{\tau_T} - \Sigma^{S,I}_{\tau_T}
      \right|
      \ge
      z_{\delta/2}
      \sqrt{T}
    \right\}
    \le
    \delta,
  \end{equation}
  where $z_{\delta/2}$ is the upper $\delta/2$-quantile of the standard Gaussian distribution
  and the inequality ``$\ge$'' in \eqref{eq:inequality} is regarded as false when $\tau_T=\infty$.
\end{corollary}

\begin{proof}
  Combine Theorem~\ref{thm:key} and Corollary~\ref{cor:Dubins-Schwarz}.
\end{proof}

\begin{corollary}[\cite{\XLIV}, Corollary~8.7]
  \label{cor:CAPM-LIL}
  Almost surely w.r.\ to $\UpProbI$,
  \begin{equation*}
    \Sigma^{S}_t\to\infty
    \Longrightarrow
    \limsup_{t\to\infty}
    \frac{\left|\Mu^S_t-\Sigma^{S,I}_t\right|}{\sqrt{2\Sigma^{S}_t\ln\ln\Sigma^{S}_t}}
    =
    1.
  \end{equation*}
\end{corollary}

\begin{proof}
  Combine Theorem~\ref{thm:key} and Corollary~\ref{cor:Dubins-Schwarz} with the law of the iterated logarithm
  for measure-theoretic Brownian motion.
\end{proof}

\Subsection{Comparisons with the standard CAPM}

  For completeness, we reproduce here Section~9 of \cite{\XLIV}
  comparing our probability-free CAPM with the standard version.
Assuming zero interested rates ($R_{{\rm f}}=0$), the standard CAPM says,
in the standard framework of measure-theoretic probability, that
\[
  \Expect(R_i)
  =
  \frac{\Cov(R_i,R_{{\rm m}})}{\Var(R_{{\rm m}})}
  \Expect(R_{{\rm m}})
\]
in the notation of \cite{Wikipedia:CAPM},
where $\Expect(R_i)$ is the expected return of the $i$th security,
$\Expect(R_{{\rm m}})$ is the expected return of the market,
$\Var(R_{{\rm m}})$ is the variance of the return of the market,
and $\Cov(R_i,R_{{\rm m}})$ is the covariance between the returns of the $i$th security and the market.

Replacing the theoretical expected values
(including those implicit in $\Var(R_{{\rm m}})$ and $\Cov(R_i,R_{{\rm m}})$)
by the empirical averages,
we obtain an approximate equality
\begin{equation}\label{eq:CAPM}
  \Mu^S_t
  \approx
  \frac{\Sigma^{S,I}_t}{\Sigma^I_t}
  \Mu^I_t.
\end{equation}
This approximate equality is still true in our probability-free framework
(under the assumptions $\Sigma^I_t\gg1$ and $\Sigma^S_t\gg1$):
indeed, our equity premium result, Corollary~\ref{cor:equity-premium}, implies $\Mu^I_t\approx\Sigma^I_t$
(e.g.,
\begin{equation*}
  \Sigma^{I}_t\to\infty
  \Longrightarrow
  \limsup_{t\to\infty}
  \frac{\left|\Mu^I_t-\Sigma^{I}_t\right|}{\sqrt{2\Sigma^{I}_t\ln\ln\Sigma^{I}_t}}
  =
  1
  \qquad
  \text{$\UpProbI$-a.s.}
\end{equation*}
is the special case of Corollary~\ref{cor:CAPM-LIL} corresponding to $S=I$),
which makes \eqref{eq:CAPM} equivalent to $\Mu^S_t\approx\Sigma^{S,I}_t$,
our game-theoretic CAPM
(see, e.g., Corollary~\ref{cor:CAPM-LIL}).
Therefore, Corollary~\ref{cor:CAPM-LIL} represents the CAPM as a law of the iterated logarithm;
similarly, Corollary~\ref{cor:CAPM-CLT} represents it as a central limit theorem.

\section{Conclusion}

This paper introduces a probability-free theory of martingales
in financial markets with continuous price paths
and applies it to the equity premium and CAPM.
These are the most obvious directions of further research:
\begin{itemize}
\item
  Allow price paths with jumps.
\item
  Explore the class of continuous martingales, as defined in this paper,
  as stochastic processes in the situation where $S_1,\ldots,S_{J^*}$
  are sample paths of measure-theoretic continuous local martingales;
  in particular, explore conditions under which this class
  coincides with the class of all continuous local martingales
  with a deterministic initial value.
  Similar questions can be asked about nonnegative supermartingales.
\end{itemize}

  \subsection*{Acknowledgments}

  This research was supported by the Air Force Office of Scientific Research
  (grant FA9550-14-1-0043).

\newcommand{\noopsort}[1]{}

\end{document}